\documentclass[english]{llncs}
\usepackage[T1]{fontenc}
\usepackage[latin9]{inputenc}
\usepackage{verbatim}
\usepackage{latexsym}
\usepackage{amsmath}
\usepackage{amssymb}
\usepackage{stmaryrd}
\usepackage{setspace}
\makeatletter

\usepackage{xcolor}
\definecolor{darkred}{rgb}{0.8,0,0} 
\usepackage{hyperref}
\hypersetup{final, colorlinks=true, linkcolor=darkred, urlcolor=darkred, citecolor=black, linkbordercolor=darkred, pdfborderstyle={/S/U/W 1}}
\usepackage{setspace}
\usepackage[absolute,overlay]{textpos} 
\usepackage{listings}
\usepackage{cite}
\usepackage{tikz}
\usetikzlibrary{automata,arrows,shapes,decorations,topaths,trees,backgrounds,shadows}

\makeatletter
\g@addto@macro \normalsize {%
 \setlength\abovedisplayskip{4pt plus 2pt minus 2pt}%
 \setlength\belowdisplayskip{4pt plus 2pt minus 2pt}%
}
\makeatother

\DeclareRobustCommand{\VAN}[2]{#2}

\makeatother

\usepackage{babel}
\begin{document}
\global\long\def\seq{\boldsymbol{x_{1}},\boldsymbol{x_{2}},...}

\setstretch{1}
\title{Convergence, Continuity and Recurrence in\\ Dynamic Epistemic Logic}
\author{Dominik Klein\inst{1}
and Rasmus K. Rendsvig\inst{2}
\institute{Department of Philosophy, Bayreuth University, and \\Department of Political Science, University of Bamberg \\ \email{dominik.klein@uni-bayreuth.de} \and Theoretical Philosophy, Lund University,  and \\ Center for Information and Bubble Studies, University of Copenhagen \\ \email{rendsvig@gmail.com}}}
\maketitle
\vspace{-10pt}
\begin{abstract}
The paper analyzes dynamic epistemic logic from a topological perspective.
The main contribution consists of a framework in which dynamic epistemic
logic satisfies the requirements for being a topological dynamical
system thus interfacing discrete dynamic logics with continuous mappings
of dynamical systems. The setting is based on a notion of logical
convergence, demonstratively equivalent with convergence in Stone
topology. Presented is a flexible, parametrized family of metrics
inducing the latter, used as an analytical aid. We show maps induced
by action model transformations continuous with respect to the Stone
topology and present results on the recurrent behavior of said maps. 


\lstset{%
language=[LaTeX]TeX,
backgroundcolor=\color{gray!25},
basicstyle=\ttfamily,
breaklines=true,
columns=fullflexible,
basewidth=0.5em,
breakindent=0pt}
\begin{textblock*}{0.35\textwidth}(160mm,50mm) 
\setstretch{.5} {
\scriptsize{} 
\noindent\texttt{Copy-Pastable BibTeX:}
}
\end{textblock*}
\begin{textblock*}{0.42\textwidth}(152mm,52mm)
\setstretch{.5} {
\scriptsize{} 
{\tiny{}\begin{lstlisting}[breaklines]
@inproceedings{KleinRendsvigLORI2017,
   author = {Klein, Dominik and Rendsvig, Rasmus K.},
    title = {{Convergence, Continuity and Recurrence
              in Dynamic Epistemic Logic}},
   editor = {Baltag, Alexandru and Seligman, Jeremy
             and Yamada, Tomoyuki}, 
booktitle = {Logic, Rationality, and Interaction
            (LORI 2017, Sapporo, Japan)},
   series = {LNCS},
     year = {2017},
publisher = {Springer}
}
\end{lstlisting}
}
\setstretch{1} }
\end{textblock*}
\end{abstract}

\section{Introduction}

Dynamic epistemic logic is a framework for modeling information dynamics.
In it, systematic change of Kripke models are punctiliously investigated
through model transformers mapping Kripke models to Kripke models.%
{} The iterated application of such a map may constitute a model of
information dynamics, or be may be analyzed purely for its mathematical
properties \cite{Baltag2009,Benthem_OneLonely,Benthem_etal_2009,Benthem2011b,Benthem-DS-2016,Bolander2015,Degremont2010,Rendsvig_BS,Rendsvig2014,Rendsvig-DS-DEL-2015,Sadzik2006}.

Dynamical systems theory is a mathematical field studying the long-term
behavior of spaces under the action of a continuous function. In case
of discrete time, this amounts to investigating the space under the
iterations of a continuous map. The field is rich in concepts, methodologies
and results developed with the aim of understanding general dynamics.

The two fields find common ground in the iterated application of maps.
With dynamic epistemic logic analyzing very specific map types, the
hope is that general results from dynamical systems theory may shed
light on properties of the former. There is, however, a chasm between
the two: Dynamical systems theory revolves around spaces imbued with
metrical or topological structure with respect to which maps are continuous.
No such structure is found in dynamic epistemic logic.  This chasm has not gone unappreciated: In his 2011 \emph{Logical
Dynamics of Information and Interaction} \cite{Benthem2011b}, van
Benthem writes
\begin{quote}
\textbf{\small{}From discrete dynamic logics to continuous dynamical
systems}{\small{} }\linebreak{}
{\small{}``We conclude with what we see as a major challenge. Van
Benthem \cite{Benthem2001,Benthem_OneLonely} pointed out how update
evolution suggests a long-term perspective that is like the evolutionary
dynamics found in dynamical systems. {[}...{]}}%
{\small{} Interfacing current dynamic and temporal logics with the
continuous realm is a major issue, also for logic in general.'' ~~~~~~~~~~~~~~~~~~~~~~~~~~~
\cite[Sec. 4.8. Emph. is org. heading]{Benthem2011b}}{\small \par}
\end{quote}
This paper takes on the challenge and attempts to bridge this chasm.\medskip{}

We proceed as follows. Section \ref{sec:Modal-Spaces} presents what
we consider natural spaces when working with modal logic, namely sets
of pointed Kripke models \emph{modulo} logical equivalence. These
are referred to as \emph{modal spaces}. A natural notion of ``logical
convergence'' on modal spaces is provided. Section \ref{sec:Topologies}
seeks a topology on modal spaces for which topological convergence
coincides with logical convergence. We consider a metric topology
based on $n$-bisimulation and prove it insufficient, but show an
adapted Stone topology satisfactory. Saddened by the loss of a useful
aid, the metric inducing the $n$-bisimulation topology, a family
of metrics is introduced that all induce the Stone topology, yet allow
a variety of subtle modelling choices. Sets of pointed Kripke models
are thus equipped with a structure of compact metric spaces. Section
\ref{sec:Clean-Maps} considers maps on modal spaces based on multi-pointed
action models using product update. Restrictions are imposed to ensure
totality, and the resulting \emph{clean maps} are shown continuous
with respect to the Stone topology. With that, we present our main
contribution: A modal space under the action of a clean map satisfies
the standard requirements for being a topological dynamical system.
Section \ref{sec:Recurrent} applies the now-suited terminology from
dynamical systems theory, and present some initial results pertaining
to the recurrent behavior of clean maps on modal spaces. Section \ref{sec:Concl}
concludes the paper by pointing out a variety of future research venues.
Throughout, we situate our work in the literature.
\begin{remark}
To make explicit what may be apparent, note that the primary concern
is the \emph{semantics} of dynamic epistemic logic, i.e., its models
and model transformation. Syntactical considerations are briefly touched
upon in Section~\ref{sec:Concl}.
\end{remark}

\begin{remark}
The paper is not self-contained. For notions from modal logic that
remain undefined here, refer to e.g. \cite{BlueModalLogic2001,ModelTheoryModalLogic}.
For topological notions, refer to e.g. \cite{Munkres}. For more on
dynamic and epistemic logic than the bare minimum of standard notions
and notations rehearsed, see e.g. \cite{BaltagMoss2004,BaltagBMS_1998,BaltagRenneReview,Baltag_Smets_2008,Benthem2011b,Ditmarsch2008,Hintikka1962,Fagin_etal_1995,Plaza1989,Rendsvig-MA-2011}.
Finally, a background document containing generalizations and omitted
proofs is our \cite{BigTech}.
\end{remark}

\section{Modal Spaces and Logical Convergence\label{sec:Modal-Spaces}}

Let there be given a countable set $\Phi$ of \textbf{atoms} and a
finite set $I$ of \textbf{agents}. Where $p\in\Phi$ and $i\in I$,
define the \textbf{language} $\mathcal{L}$ by
\begin{center}
$\varphi:=\top\;|\;p\;|\;\neg\varphi\;|\;\varphi\wedge\varphi\;|\;\square_{i}\varphi$.
\par\end{center}

Modal logics may be formulated in $\mathcal{L}$. By a \textbf{logic
}$\Lambda$ we refer only to extensions of the \textbf{minimal normal
modal logic} $K$ over the language $\mathcal{L}$. With $\Lambda$
given by context, let $\boldsymbol{\varphi}$ be the set of formulas
$\Lambda$-provably equivalent to $\varphi$. Denote the resulting
partition $\{\boldsymbol{\varphi}:\varphi\in\mathcal{L}\}$ of $\mathcal{L}$
by $\boldsymbol{\mathcal{L}}_{\Lambda}$.\footnote{$\boldsymbol{\mathcal{L}_{\Lambda}}$ is isomorphic to the domain
of the \textbf{Lindenbaum algebra} of $\Lambda$.} Call $\boldsymbol{\mathcal{L}}_{\Lambda}$'s elements $\Lambda$\textbf{-propositions}.

We use relational semantics to evaluate formulas. A \textbf{Kripke
model} for $\mathcal{L}$ is a tuple $M=(\left\llbracket M\right\rrbracket ,R,\left\llbracket \cdot\right\rrbracket )$
where $\left\llbracket M\right\rrbracket $ is a countable, non-empty
set of \textbf{states}, $R:I\longrightarrow\mathcal{P}(\left\llbracket M\right\rrbracket \times\left\llbracket M\right\rrbracket )$
assigns to each $i\in I$ an \textbf{accessibility relation} $R_{i}$,
and $\left\llbracket \cdot\right\rrbracket :\Phi\longrightarrow\mathcal{P}(\left\llbracket M\right\rrbracket )$
is a \textbf{valuation}, assigning to each atom a set of states. With
$s\in\left\llbracket M\right\rrbracket $, call $Ms=(\left\llbracket M\right\rrbracket ,R,\left\llbracket \cdot\right\rrbracket ,s)$
a \textbf{pointed Kripke model}. The used semantics are standard,
including the modal clause:
\[
Ms\vDash\square_{i}\varphi\mbox{ iff for all }t:sR_{i}t\mbox{ implies }Mt\vDash\varphi.
\]

\noindent Throughout, we work with pointed Kripke models. Working
with modal logics, we find it natural to identify pointed Kripke models
that are considered equivalent by the logic used. The domains of interest
are thus the following type of quotient spaces:
\begin{definition}
The $\boldsymbol{\mathcal{L}}_{\Lambda}$ \textbf{modal space} of
a set of pointed Kripke models $X$ is the set $\boldsymbol{X}=\{\boldsymbol{x}:x\in X\}$
for $\boldsymbol{x}=\{y\in X:y\vDash\varphi\text{ iff }x\vDash\varphi\text{ for all }\boldsymbol{\varphi}\in\boldsymbol{\mathcal{L}}_{\Lambda}\}$.
\end{definition}

\noindent Working with an $\boldsymbol{\mathcal{L}}_{\Lambda}$ modal
space portrays that we only are interested in differences between
pointed Kripke models insofar as these are modally expressible and
are considered differences by $\Lambda$.\smallskip{}

In a modal space, how may we conceptualize that a sequence $\boldsymbol{x_{1}},\boldsymbol{x_{2}},...$
converges to some point $\boldsymbol{x}$? Focusing on the concept
from which we derive the notion of identity in modal spaces, namely
$\Lambda$-propositions, we find it natural to think of $\seq$ as
converging to $\boldsymbol{x}$ just in case $\boldsymbol{x_{n}}$
moves towards satisfying all the same $\Lambda$-propositions as $\boldsymbol{x}$
as $n$ goes to infinity. We thus offer the following definition:
\begin{definition}
A sequence of points $\seq$ in an $\boldsymbol{\mathcal{L}}_{\Lambda}$
modal space $\boldsymbol{X}$ is said to \textbf{logically converge}
to the point $\boldsymbol{x}$ in $\boldsymbol{X}$ iff for every
$\boldsymbol{\varphi}\in\boldsymbol{\mathcal{L}}_{\Lambda}$ for which
$x\vDash\varphi$, there is an $N\in\mathbb{N}$ such that $x_{n}\vDash\varphi$
for all $n\geq N$.
\end{definition}

To avoid re-proving useful results concerning this notion of convergence,
we next turn to seeking a topology for which logical convergence coincides
with topological convergence. Recall that for a topology $\mathcal{T}$
on a set $X$, a sequence of points $x_{1},x_{2},...$ is said to
\textbf{converge} to $x$ in the \textbf{topological space} $(X,\mathcal{T})$
iff for every open set $U\in\mathcal{T}$ containing $x$, there is
an $N\in\mathbb{N}$ such that $x_{n}\in U$ for all $n\geq N$. 

\section{Topologies on Modal Spaces\label{sec:Topologies}}

One way of obtaining a \textbf{topology} on a space is to define a
\textbf{metric} for said space. Several metrics have been suggested
for sets of pointed Kripke models \cite{aucher2010generalizing,Caridroit2016}.
These metrics are only defined for finite pointed Kripke models, but
incorporating ideas from the metrics of \cite{Lind1995} on \textbf{shift
spaces} and \cite{Goranko2004} on sets of \textbf{first-order logical
theories} allows us to simultaneously generalize and simplify the
\textbf{$n$-Bisimulation-based Distance} of \cite{Caridroit2016}
to the degree of applicability:

Let $\boldsymbol{X}$ be a modal space for which modal equivalence
and \textbf{bisimilarity} coincide\footnote{That all models in $X$ are \textbf{image-finite} is a sufficient
condition, cf. the Hennessy-Milner Theorem. See e.g. \cite{BlueModalLogic2001}
or \cite{ModelTheoryModalLogic}.} and let $\leftrightarroweq_{n}$ relate $x,y\in X$ iff $x$ and
$y$ are \textbf{$n$-bisimilar}. Then proving 

\setstretch{.8}
\[
d_{B}(\boldsymbol{x},\boldsymbol{y})=\begin{cases}
0 & \text{ if }x\leftrightarroweq_{n}y\text{ for all }n\\
\frac{1}{2^{n}} & \text{ if }n\text{ is the least intenger such that }x\not\leftrightarroweq_{n}y
\end{cases}
\]
\setstretch{1}a metric on $\boldsymbol{X}$ is trivial. We refer
to $d_{B}$ as the \textbf{$n$-bisimulation metric}, and to the induced
\textbf{metric topology }as the \textbf{$n$-bisimulation topology},
denoted $\mathcal{T}_{B}$. A basis of the topology $\mathcal{T}_{B}$
is given by the set of elements ${B_{\boldsymbol{x}n}=\{\boldsymbol{y}\in\boldsymbol{X}\colon y\leftrightarroweq_{n}x\}.}$\smallskip{}

Considering the intimate link between modal logic and bisimulation,
we consider both $n$-bisimulation metric and topology highly natural.\footnote{Space does not allow for a discussion of the remaining metrics of
\cite{aucher2010generalizing,Caridroit2016}, but see \cite{BigTech}.} Alas, logical convergence does not:
\begin{proposition}
\label{prop:suck-it}Logical convergence in arbitrary modal space
$\boldsymbol{X}$ does not imply convergence in the topological space
$(\boldsymbol{X},\mathcal{T}_{B})$.
\end{proposition}

\begin{proof}
Let $\boldsymbol{X}$ be an $\boldsymbol{\mathcal{L}}_{\Lambda}$
modal space with $\mathcal{L}$ based on the atoms $\Phi=\{p_{k}\colon k\in\mathbb{N}\}$.
Let $x\in X$ satisfy $\square\bot$ and $p_{k}$ for all $k\in\mathbb{N}$.
Let \textbf{$\boldsymbol{x_{1}},\boldsymbol{x_{2}},...$} be a sequence
in $\boldsymbol{X}$ such that for all $k\in\mathbb{N}$, $x_{k}$
satisfies $\square\bot$, $p_{m}$ for all $m\leq k$, and $\neg p_{l}$
for all $l>k$. Then for all $\boldsymbol{\varphi}\in\boldsymbol{\mathcal{L}}_{\Lambda}$
for which $x\vDash\varphi$, there is an $N$ such that $x_{n}\vDash\varphi$
for all $n\geq N$, hence the sequence \textbf{$\boldsymbol{x_{1}},\boldsymbol{x_{2}},...$}
converges to $\boldsymbol{x}$.There does not, however, exist any
$N'$ such that $\boldsymbol{x}_{n'}\in B_{\boldsymbol{x}0}$ for
all $n'\geq N'$. Hence $\boldsymbol{x_{1}},\boldsymbol{x_{2}},...$
does not converge to $\boldsymbol{x}$ in $\mathcal{T}_{B}$.\qed
\end{proof}

\noindent Proposition \ref{prop:suck-it} implies that the $n$-bisimulation
topology may not straight-forwardly be used to establish negative
results concerning logical convergence. That it may be used for positive
cases is a corollary to Propositions \ref{prop:yay} and \ref{prop:nBisim-strictly-finer}
below. On the upside, logical convergence coincides with convergence
in the \textbf{$n$}-bisimulation topology \textendash{} i.e. Proposition
\ref{prop:suck-it} fails \textendash{} when $\mathcal{L}$ has finite
atoms. This is a corollary to Proposition  \ref{prop:nBisim-finite}.\medskip{}

An alternative to a metric-based approach to topologies is to construct
the set of all open sets directly. Comparing the definition of logical
convergence with that of convergence in topological spaces is highly
suggestive: Replacing every occurrence of the formula $\varphi$ with
an open set $U$ while replacing satisfaction $\vDash$ with inclusion
$\in$ transforms the former definition into the latter. Hence the
collection of sets $U_{\boldsymbol{\varphi}}=\{\boldsymbol{x}\in\boldsymbol{X}:x\vDash\varphi\}$,
$\boldsymbol{\varphi}\in\boldsymbol{\mathcal{L}}_{\Lambda}$, seems
a reasonable candidate for a topology. Alas, this collection is not
closed under arbitrary unions, as all formulas are finite. Hence it
is not a topology. It does however constitute the basis for a topology,
in fact the somewhat influential \textbf{Stone topology}, $\mathcal{T}_{S}$.

The Stone topology is traditionally defined on the collection of complete
theories for some propositional, first-order or modal logic, but is
straightforwardly applicable to modal spaces. Moreover, it satisfies
our \emph{desideratum}:
\begin{proposition}
\label{prop:yay}For any \textup{$\boldsymbol{\mathcal{L}}_{\Lambda}$}
modal space $\boldsymbol{X}$, a sequence $\seq$ logically converges
to the point $\boldsymbol{x}$ if, and only if, it converges to $\boldsymbol{x}$
in $(\boldsymbol{X},\mathcal{T}_{S})$.
\end{proposition}

\begin{proof}
Assume $\boldsymbol{x_{1}},\boldsymbol{x_{2}},...$ logically converges
to $\boldsymbol{x}$ in $\boldsymbol{X}$ and that $U$ containing
$\boldsymbol{x}$ is open in $\mathcal{T}_{S}$. Then there is a basis
element $U_{\boldsymbol{\varphi}}\subseteq U$ with $\boldsymbol{x}\in U_{\boldsymbol{\varphi}}$.
So $x\vDash\varphi$. By assumption, there exists an $N$ such that
$x_{n}\vDash\varphi$ for all $n\geq N$. Hence $\boldsymbol{x_{n}}\in U_{\boldsymbol{\varphi}}\subseteq U$
for all $n\geq N$.

Assume $\boldsymbol{x_{1}},\boldsymbol{x_{2}},...$ converges to $\boldsymbol{x}$
in $(\boldsymbol{X},\mathcal{T}_{S})$ and let $x\vDash\varphi$.
Then $\boldsymbol{x}\in U_{\boldsymbol{\varphi}}$, which is open.
As the sequence converges, there exists an $N$ such that $\boldsymbol{x_{n}}\in U_{\boldsymbol{\varphi}}$
for all $n\geq N$. Hence $x_{n}\vDash\varphi$ for all $n\geq N$.\qed
\end{proof}

Apart from its attractive characteristic concerning convergence, working
on the basis of a logic, the Stone topology imposes a natural structure.
As is evident from its basis, every subset of $\boldsymbol{X}$ characterizable
by a single $\Lambda$-proposition $\boldsymbol{\varphi}\in\boldsymbol{\mathcal{L}}_{\Lambda}$
is clopen. If the logic $\Lambda$ is compact and $\boldsymbol{X}$
saturated (see fn. \ref{fn:sat}), also the converse is true: every
clopen set is of the form $U_{\varphi}$ for some $\varphi$. We refer
to \cite{BigTech} for proofs and a precise characterization result.
In this case, a subset is open, but not closed, iff it is characterizable
only by an infinitary disjunction of $\Lambda$-propositions, and
a subset if closed, but not open, iff it is characterizable only by
an infinitary conjunction of $\Lambda$-propositions. The Stone topology
thus transparently reflects the properties of logic, language and
topology. Moreover, it enjoys practical topological properties:
\begin{proposition}
\label{prop:Stone}For any \textup{$\boldsymbol{\mathcal{L}}_{\Lambda}$}
modal space $\boldsymbol{X}$, $(\boldsymbol{X},\mathcal{T}_{S})$
is \textbf{Hausdorff} and \textbf{totally disconnected}. If $\Lambda$
is (logically)\textbf{ compact}\footnote{\noindent A logic $\Lambda$ is logically compact if any arbitrary
set $A$ of formulas is $\Lambda$-consistent iff every finite subset
of $A$ is $\Lambda$-consistent.} and $\boldsymbol{X}$ is \textbf{saturated}\footnote{\noindent \label{fn:sat}An $\boldsymbol{\mathcal{L}}_{\Lambda}$
modal space $\boldsymbol{X}$ is saturated iff for each $\Lambda$-consistent
set of formulas $A$, there is an $\boldsymbol{x}\in\boldsymbol{X}$
such that $x\vDash A$. Saturation relates to the notion of \textbf{strong
completeness}, cf. e.g. \cite[Prop.4.12]{BlueModalLogic2001}. See
\cite{BigTech} for its use in a more general context.}, then $(\boldsymbol{X},\mathcal{T}_{S})$ is also (topologically)\textbf{
compact}. 
\end{proposition}

\begin{proof}
\noindent These properties are well-known for the Stone topology applied
to complete theories. For the topology applied to modal spaces, we
defer to \cite{BigTech}.\qed
\end{proof}

One may interject that, as having a metric may facilitate obtaining
results, it may cause a loss of tools to move away from the $n$-bisimulation
topology. The Stone topology, however, is \textbf{metrizable}. A family
of metrics inducing it, generalizing the Hamming distance to infinite
strings by using weighted sums, was introduced in \cite{BigTech}.
We here present a sub-family, suited for modal spaces:
\begin{definition}
\label{def:metrics}Let $D\subseteq\boldsymbol{\mathcal{L}}_{\Lambda}$
contain for every $\boldsymbol{\psi}\in\boldsymbol{\mathcal{L}}_{\Lambda}$
some $\{\boldsymbol{\boldsymbol{\varphi}}_{i}\}_{i\in I}$ that $\Lambda$-entails
either $\mathbf{\boldsymbol{\psi}}$ or $\mathbf{\boldsymbol{\neg\psi}}$,
and let \textbf{$\boldsymbol{\varphi}_{1},\boldsymbol{\varphi}_{2},...$}
be an enumeration of $D$. 

Let $\boldsymbol{X}$ be an $\boldsymbol{\mathcal{L}}_{\Lambda}$
modal space. For all $\boldsymbol{x},\boldsymbol{y}\in\boldsymbol{X}$,
for all $k\in\mathbb{N}$, let

\setstretch{.8}\vspace{4pt}

\noindent 
\[
d_{k}(\boldsymbol{x},\boldsymbol{y})={\textstyle \begin{cases}
0 & \text{ if }x\vDash\varphi\text{ iff }y\vDash\varphi\text{ for \ensuremath{\varphi\in\boldsymbol{\varphi}_{k}}}\\
1 & \text{ else}
\end{cases}}
\]

\setstretch{1}Let $w:D\longrightarrow\mathbb{R}_{>0}$ assign strictly
positive \textbf{weight} to each $\boldsymbol{\varphi}_{k}$ in $D$
such that $(w(\boldsymbol{\varphi}_{n}))$ forms a convergent series.
Define the function $d_{w}:\boldsymbol{X}^{2}\longrightarrow\mathbb{R}$
by
\[
d_{w}(\boldsymbol{x},\boldsymbol{y})=\sum_{k=0}^{\infty}w(\boldsymbol{\varphi}_{k})d_{k}(\boldsymbol{x},\boldsymbol{y})
\]
for all $\boldsymbol{x},\boldsymbol{y}\in\boldsymbol{X}$. The set
of these functions is denoted $D_{\boldsymbol{X}}$. Let $\mathcal{D}_{D,\boldsymbol{X}}=\cup_{D\subseteq\boldsymbol{\mathcal{L}}_{\Lambda}}D_{\boldsymbol{X}}$.
\end{definition}

We refer to \cite{BigTech} for the proof establishing the following
proposition:
\begin{proposition}
\label{prop:metricsStone}Let $\boldsymbol{X}$ be an \textup{$\boldsymbol{\mathcal{L}}_{\Lambda}$}
modal space and $d_{w}$ belong to $\mathcal{D}_{\boldsymbol{X}}$.
Then $d_{w}$ is a metric on $\boldsymbol{X}$ and the metric topology
$\mathcal{T}_{w}$ induced by $d_{w}$ on $\boldsymbol{X}$ is the
Stone topology of $\Lambda$.
\end{proposition}

\noindent For a metric space $(\boldsymbol{X},d)$, we will also write
$\boldsymbol{X}_{d}$.\medskip{}

With variable parameters $D$ and $w$, $\mathcal{D}_{\boldsymbol{X}}$
allows one to vary the choice of metric with the problem under consideration.
E.g., if the $n$-bisimulation metric seems apt, one could choose
that, with one restriction:
\begin{proposition}
\label{prop:nBisim-finite}If $\boldsymbol{X}$ is an \textup{$\boldsymbol{\mathcal{L}}_{\Lambda}$}
modal space with $\mathcal{L}$ based on a finite atom set%
, then $\mathcal{D}_{\boldsymbol{X}}$ contains a topological equivalent
to the $n$-bisimulation metric.
\end{proposition}

\begin{proof}
[sketch] As $\mathcal{L}$ is based on a finite set of atoms, for
each $\boldsymbol{x}\in\boldsymbol{X},n\in\mathbb{N}_{0}$, there
exists a characteristic formula $\varphi_{x,n}$ such that $y\vDash\varphi_{x,n}$
iff $y\leftrightarroweq_{n}x$, cf. \cite{ModelTheoryModalLogic}.
Let $D_{n}=\{\boldsymbol{\varphi_{x,n}}\colon\boldsymbol{x}\in\boldsymbol{X}\}$
and $D=\cup_{n\in\mathbb{N}_{0}}D_{n}$. Then each $D_{n}$ is finite
and $D$ satisfies Definition \ref{def:metrics}. Finally, let $w(\boldsymbol{\varphi})=\frac{1}{|D_{n}|}\cdot\frac{1}{2^{n+1}}$
for $\boldsymbol{\varphi}\in D_{n}$. Then $d_{w}\in\mathcal{D}_{\boldsymbol{X}}$
and is equivalent to the $n$-bisimulation metric $d_{b}$.\qed%
\end{proof}

\noindent As a corollary to Proposition \ref{prop:nBisim-finite},
it follows that, for finite atom languages, the $n$-bisimulation
topology is the Stone topology. This is not true in general, as witnessed
by Proposition \ref{prop:suck-it} and the following:
\begin{proposition}
\label{prop:nBisim-strictly-finer}If $\boldsymbol{X}$ is an \textup{$\boldsymbol{\mathcal{L}}_{\Lambda}$}
modal space with $\mathcal{L}$ based on a countably infinite atom
set, then the $n$-bisimulation metric topology on $\boldsymbol{X}$
is strictly finer than the Stone topology on $\boldsymbol{X}$.
\end{proposition}

\begin{proof}
[sketch]We refer to \cite{BigTech} for details, but for $\mathcal{T}_{B}\not\subseteq\mathcal{T}_{S}$,
note that the set $B_{\boldsymbol{x}0}$ used in the proof of Prop.
\ref{prop:suck-it}, is open in $\mathcal{T}_{B}$, but not in $\mathcal{T}_{S}$.\qed
\end{proof}

With this comparison, we end our exposition of topologies on modal
spaces.

\section{Clean Maps on Modal Spaces\label{sec:Clean-Maps}}

We focus on a class of maps induced by action models applied using
product update. Action models are a popular and widely applicable
class of model transformers, generalizing important constructions
such as public announcements. An especially general version of action
models is \emph{multi-pointed} action models with \emph{postconditions}.
Postconditions allow action states in an action model to change the
valuation of atoms \cite{Benthem2006_com-change,Ditmarsch_Kooi_ontic},
thereby also allowing the representation of information dynamics concerning
situations that are not factually static. Permitting multiple points
allows the actual action states executed to depend on the pointed
Kripke model to be transformed, thus generalizing single-pointed action
models.\footnote{Multi-pointed action models are also referred to as \emph{epistemic
programs} in \cite{BaltagMoss2004}, and allow encodings akin to \emph{knowledge-based
programs} \cite{Fagin_etal_1995} of interpreted systems, cf. \cite{Rendsvig-DS-DEL-2015}.}

A \textbf{multi-pointed action model} is a tuple $\Sigma{\scriptstyle \Gamma}=(\llbracket\Sigma\rrbracket,\mathsf{R},pre,post,\Gamma)$
where $\left\llbracket \Sigma\right\rrbracket $ is a countable, non-empty
set of \textbf{actions}. The map $\mathsf{R}:I\rightarrow\mathcal{P}(\left\llbracket \Sigma\right\rrbracket \times\left\llbracket \Sigma\right\rrbracket )$
assigns an \textbf{accessibility relation} $\mathsf{R}_{i}$ on $\llbracket\Sigma\rrbracket$
to each agent $i\in I$. The map ${pre:\left\llbracket \Sigma\right\rrbracket \rightarrow\mathcal{L}}$
assigns to each action a \textbf{precondition}, and the map \linebreak{}
${post:\left\llbracket \Sigma\right\rrbracket \rightarrow\mathcal{L}}$
assigns to each action a \textbf{postcondition},\footnote{The precondition of $\sigma$ specify the conditions under which $\sigma$
is executable, while its postcondition may dictate the posterior values
of a finite, possibly empty, set of atoms.} which must be $\top$ or a conjunctive clause\footnote{I.e. a conjuction of literals, where a literal is an atom or a negated
atom.} over $\Phi$. Finally, $\emptyset\not=\Gamma\subseteq\left\llbracket \Sigma\right\rrbracket $
is the set of designated~actions. 

To obtain well-behaved total maps on a modal spaces, we must invoke
a set of mild, but non-standard, requirements: Let $X$ be a set of
pointed Kripke models. Call $\Sigma{\scriptstyle \Gamma}$ \textbf{precondition
finite} if the set $\{\boldsymbol{pre(\sigma)}\in\boldsymbol{\mathcal{L}}_{\Lambda}\colon\sigma\in\left\llbracket \Sigma\right\rrbracket \}$
is finite. This is needed for our proof of continuity. Call $\Sigma{\scriptstyle \Gamma}$
\textbf{exhaustive over }$X$ if for all $x\in X$, there is a $\sigma\in\Gamma$
such that $x\vDash pre(\sigma)$. This conditions ensures that the
action model $\Sigma{\scriptstyle \Gamma}$ is universally applicable
on $X$. Finally, call $\Sigma{\scriptstyle \Gamma}$ \textbf{deterministic
over }$X$ if $X\vDash pre(\sigma)\wedge pre(\sigma')\rightarrow\bot$
for each $\sigma\neq\sigma'\in\Gamma$. Together with exhaustivity,
this condition ensures that the product of $\Sigma{\scriptstyle \Gamma}$
and any $Ms\in X$ is a (single-)pointed Kripke model, i.e., that
the actual state after the updates is well-defined and unique.

Let $\Sigma{\scriptstyle \Gamma}$ be exhaustive and deterministic
over $X$ and let $Ms\in X$. Then the \textbf{product update} of
$Ms$ with $\Sigma{\scriptstyle \Gamma}$, denoted $Ms\otimes\Sigma{\scriptstyle \Gamma}$,
is the pointed Kripke model $(\left\llbracket M\Sigma\right\rrbracket ,R',\llbracket\cdot\rrbracket',s')$
with
\begin{eqnarray*}
\left\llbracket M\Sigma\right\rrbracket  & = & \left\{ (s,\sigma)\in\left\llbracket M\right\rrbracket \times\left\llbracket \Sigma\right\rrbracket :(M,s)\vDash pre(\sigma)\right\} \\
R' & = & \left\{ ((s,\sigma),(t,\tau)):(s,t)\in R_{i}\mbox{ and }(\sigma,\tau)\in\mathsf{R}_{i}\right\} ,\text{ for all }i\in N\\
\left\llbracket p\right\rrbracket ' & = & \left\{ (s,\sigma)\!:\!s\in\left\llbracket p\right\rrbracket \!,post(\sigma)\nvDash\neg p\right\} \cup\left\{ (s,\sigma)\!:\!post(\sigma)\vDash p\right\} ,\text{ for all }p\in\Phi\\
s' & = & (s,\sigma):\sigma\in\Gamma\mbox{ and }Ms\vDash pre(\sigma)
\end{eqnarray*}

\noindent Call $\Sigma{\scriptstyle \Gamma}$ \textbf{closing over
$X$} if for all $x\in X,$ $x\otimes\Sigma{\scriptstyle \Gamma}\in X$.
With $\Sigma{\scriptstyle \Gamma}$ exhaustive and deterministic, $\Sigma{\scriptstyle \Gamma}$ and $\otimes$ induce a well-defined total map on $X$.\smallskip{}

The class of maps of interest in the present is then the following:
\begin{definition}
Let $\boldsymbol{X}$ be an \textup{$\boldsymbol{\mathcal{L}}_{\Lambda}$}
modal space. A map $\boldsymbol{f}:\boldsymbol{X}\rightarrow\boldsymbol{X}$
is called \textbf{clean} \linebreak{}
if there exists a precondition finite, multi-pointed action model
$\Sigma{\scriptstyle \Gamma}$ closing, deterministic and exhaustive
over $X$ such that $\boldsymbol{f}(\boldsymbol{x})=\boldsymbol{y}$
iff $x\otimes\Sigma{\scriptstyle \Gamma}\in\boldsymbol{y}$ for all
$\boldsymbol{x}\in\boldsymbol{X}$.
\end{definition}

\noindent Clean maps are total by the assumptions of being closing
and exhaustive. They are well-defined as $\boldsymbol{f}(\boldsymbol{x})$
is independent of the choice of representative for $\boldsymbol{x}$:
If $x'\in\boldsymbol{x}$, then $x'\otimes\Sigma{\scriptstyle \Gamma}$
and $x\otimes\Sigma{\scriptstyle \Gamma}$ are modally equivalent
and hence define the same point in $\mathbf{X}$. The latter follows
as multi-pointed action models applied using product update preserve
bisimulation \cite{BaltagMoss2004}, which implies modal equivalence.
Clean maps moreover play nicely with the Stone topology:
\begin{proposition}
\label{prop:cont}Let $\boldsymbol{f}$ be a clean map on an \textup{$\boldsymbol{\mathcal{L}}_{\Lambda}$}
modal space $\boldsymbol{X}$. Then $\boldsymbol{f}$ is continuous
with respect to the Stone topology of $\Lambda$.
\end{proposition}

\begin{proof}
[sketch] We defer to \cite{BigTech} for details, but offer a sketch:
The map $\boldsymbol{f}$ is shown uniformly continuous using the
$\varepsilon$-$\delta$ formulation of continuity. The proof relies
on a lemma stating that for every $d_{w}\in\mathcal{D}_{\boldsymbol{X}}$
and every $\epsilon>0$, there are formulas $\chi_{1},\ldots,\chi_{l}\in\mathcal{L}$
such that every $x\in X$ satisfies some $\chi_{i}$ and whenever
$y\vDash\chi_{i}$ and $z\vDash\chi_{i}$ for some $i\leq l$, then
$d_{w}(y,z)<\epsilon$. The main part of the proof establishes the
claim that there is a function $\boldsymbol{\delta}:\mathcal{L}\rightarrow(0,\infty)$
such that for any $\varphi\in\mathcal{L}$, if $f(x)\vDash\varphi$
and $d_{\vec{a}}(x,y)<\boldsymbol{\delta}(\varphi)$, then $f(y)\vDash\varphi$.
Setting $\delta=\min\{\boldsymbol{\delta}(\chi_{i})\colon i\leq l\}$
then yields a $\delta$ with the desired property.\qed
\end{proof}

With Proposition \ref{prop:cont}, we are positioned to state our
main theorem:
\begin{theorem}
\label{thm:DynSys}Let $\boldsymbol{f}$ be a clean map on a saturated
\textup{\emph{$\boldsymbol{\mathcal{L}}_{\Lambda}$ modal space }}$\boldsymbol{X}$
with $\Lambda$ compact and let $d\in\mathcal{D}_{\boldsymbol{X}}$.
Then $(\boldsymbol{X}_{d},\boldsymbol{f})$ is a \textbf{topological
dynamical system}.
\end{theorem}

\begin{proof}
Propositions \ref{prop:yay}, \ref{prop:Stone}, \ref{prop:metricsStone}
and \ref{prop:cont} jointly imply that $\boldsymbol{X}_{d}$ is a
compact metric space on which $\boldsymbol{f}$ is continuous, thus
satisfying the requirements of e.g. \cite{Hassellblatt2002,Vries2014,Eisner2015}.\qed
\end{proof}

With Theorem \ref{thm:DynSys}, we have, in what we consider a natural
manner, situated dynamic epistemic logic in the mathematical discipline
of dynamical systems. A core topic in this discipline is to understand
the long-term, qualitative behavior of maps on spaces. Central to
this endeavor is the concept of \emph{recurrence}, i.e., understanding
when a system returns to previous states as time goes to infinity.

\section{Recurrence in the Limit Behavior of Clean Maps\label{sec:Recurrent}}

We represent results concerning the limit behavior of clean maps on
modal spaces. In establishing the required terminology, we follow
\cite{Hassellblatt2002}: Let $f$ be a continuous map on a metric
space $X_{d}$ and $x\in X_{d}$. A point $y\in X$ is a \textbf{limit
point}\footnote{Or $\omega$-limit point. The $\omega$ is everywhere omitted as time
here only moves forward.} for $x$ under $f$ if there is a strictly increasing sequence $n_{1},n_{2},...$
such that the subsequence $f^{n_{1}}(x),f^{n_{2}}(x),...$ of $(f^{n}(x))_{n\in\mathbb{N}_{0}}$
converges to $y$. The \textbf{limit set }of $x$ under $f$ is the
set of all limit points for $x$, denoted $\omega_{f}(x)$. Notably,
$\omega_{f}(x)$ is closed under $f$: For $y\in\omega_{f}(x)$ also
$f(y)\in\omega_{f}(x)$. We immediately obtain that any modal system
satisfying Theorem \ref{thm:DynSys} has a nonempty limit set:
\begin{proposition}
\label{prop:Comp->LimitSet} Let $(\boldsymbol{X}_{d},\boldsymbol{f})$
be as in Theorem \ref{thm:DynSys}. For any point $\boldsymbol{x}\in\boldsymbol{X}$,
the limit set of $\boldsymbol{x}$ under $\boldsymbol{f}$ is non-empty.
\end{proposition}

\begin{proof}
Since $\boldsymbol{X}$ is is compact, every sequence in $\boldsymbol{X}$
has a convergent subsequence, cf. e.g. \cite[Thm. 28.2]{Munkres}.
\end{proof}

Proposition \ref{prop:Comp->LimitSet} does not inform us of the \emph{structure
}of said limit set. In the study of dynamical systems, such structure
is often sought through classifying the possible repetitive behavior
of a system, i.e., through the system's \emph{recurrence} properties.
For such studies, a point $x$ is called (positively) \textbf{recurrent
}if ${x\in\omega_{f}(x)}$, i.e., if it is a limit point of itself.
\medskip{}

The simplest structural form of recurrence is \emph{periodicity}:
For a point $x\in X$, call the set $\mathcal{O}_{f}(x)=\{f^{n}(x)\colon n\in\mathbb{N}_{0}\}$
its \textbf{orbit}. The orbit $\mathcal{O}_{f}(x)$ is \textbf{periodic}
if $f^{n+k}(x)=f^{n}(x)$ for some $n\geq0,k>0$; the least such $k$
is the \textbf{period} of $\mathcal{O}_{f}(x)$. Periodicity is thus
equivalent to $\mathcal{O}_{f}(x)$ being finite. Related is the notion
of a \textbf{limit cycle}:\textbf{ }a periodic orbit $\mathcal{O}_{f}(x)$
is a limit cycle if it is the limit set of some $y$ not in the period,
i.e., if $\mathcal{O}_{f}(x)=\omega_{f}(y)$ for some $y\not\in\mathcal{O}_{f}(x)$. 

\medskip{}

It was conjectured by van Benthem that certain clean maps\textemdash those
based on finite action models and without postconditions\textemdash would,
whenever applied to a finite $x$, have a periodic orbit $\mathcal{O}_{f}(x)$.
I.e., after finite iterations, the map would oscillate between a finite
number of states. This was the content of van Benthem's ``\emph{Finite
Evolution Conjecture}'' \cite{Benthem_OneLonely}\emph{. }The conjecture
was refuted using a counterexample by Sadzik in his 2006 paper, \cite{Sadzik2006}.\footnote{We paraphrase van Benthem and Sadzik using the terminology introduced.}
The example provided by Sadzik (his Example 33) uses an action model
with only Boolean preconditions. Interestingly, the orbit of the corresponding
clean map terminates in a limit cycle. This is a corollary to Proposition
\ref{Niceprop} below.\medskip{}

Before we can state the proposition, we need to introduce some terminology.
Call a multi-pointed action model $\Sigma{\scriptstyle \Gamma}$ \textbf{finite}
if $\left\llbracket \Sigma\right\rrbracket $ is finite, \textbf{Boolean
}if $pre(\sigma)$ is a Boolean formula for all $\sigma\in\left\llbracket \Sigma\right\rrbracket $,
and \textbf{static} if $post(\sigma)=\top$ for all $\sigma\in\left\llbracket \Sigma\right\rrbracket $.
We apply the same terms to a clean map $\boldsymbol{f}$ based on
$\Sigma{\scriptstyle \Gamma}$. In this terminology, Sadzik showed
that for any finite, Boolean, and static clean map $\boldsymbol{f}\colon\boldsymbol{X}\rightarrow\boldsymbol{X}$,
if the orbit $\mathcal{O}_{\boldsymbol{f}}(\boldsymbol{x})$ is periodic,
then it has period $1$.\footnote{See \cite{Bolander2015} for an elegant and generalizing exposition.}
This insightful result immediates the following:
\begin{proposition}
\label{Niceprop} Let $(\boldsymbol{X}_{d},\boldsymbol{f})$ be as
in Theorem \ref{thm:DynSys} with $\boldsymbol{f}$ finite, Boolean,
and static. For all $\boldsymbol{x}\in\boldsymbol{X}$, the orbit
$\mathcal{O}_{\boldsymbol{f}}(\boldsymbol{x})$ is periodic with period
$1$.
\end{proposition}

\begin{proof}
By Prop. \ref{prop:Comp->LimitSet}, the limit set $\omega_{\boldsymbol{f}}(\boldsymbol{x})$
of $\boldsymbol{x}$ under $\boldsymbol{f}$ is non-empty. Sadzik's
result shows that it contains a single point. Hence $(\boldsymbol{f}^{n}(\boldsymbol{x}))_{n\in\mathbb{N}_{0}}$
converges to this point. As the limit set $\omega_{\boldsymbol{f}}(\boldsymbol{x})$
is closed under $\boldsymbol{f}$, its unique point is a fix-point.\qed
\end{proof}

Proposition \ref{Niceprop} may be seen as a partial vindication of
van Benthem's conjecture: Forgoing the requirement of reaching the
limit set in finite time and the possibility of modal preconditions,
the conjecture holds, even if the initial state has an infinite set
of worlds $\left\llbracket x\right\rrbracket $. This simple recurrent
behavior is, however, not the general case. More complex clean maps
may exhibit \textbf{nontrivial recurrence}, i.e., produce non-periodic
orbits with recurrent points:
\begin{proposition}
\label{prop:NonBool-NoPost}There exist finite, static, but non-Boolean,
clean maps that exhibit nontrivial recurrence.
\end{proposition}

\begin{proposition}
\label{prop:Bool-Post}There exist finite, Boolean, but non-static,
clean maps that exhibit nontrivial recurrence.
\end{proposition}

\noindent We show these propositions below, building a clean map which,
from a selected initial state, has uncountably many limit points,
despite the orbit being only countable. Moreover, said orbit also
contains infinitely many recurrent points. In fact, every element
of the orbit is recurrent. We rely on Lemma \ref{lem:Any-Turing}
in the proof. A proof of Lemma \ref{lem:Any-Turing}.1 may be found
in \cite{KleinRendsvigTuring}, a proof of Lemma \ref{lem:Any-Turing}.2
in \cite{BolanderBirkegaard2011}.
\begin{lemma}
\label{lem:Any-Turing}Any Turing machine can be emulated using a
set $X$ of $S5$ pointed Kripke models for finite atoms and a finite
multi-pointed action model $\Sigma{\scriptstyle \Gamma}$ deterministic
over $X$. Moreover, $\Sigma{\scriptstyle \Gamma}$ may be chosen
1. static, but non-Boolean, or 2. Boolean, but non-static.
\end{lemma}

\begin{proof}
[of Propostitions \ref{prop:NonBool-NoPost} and \ref{prop:Bool-Post}]
For both propositions, we use a Turing machine \emph{ad infinitum
}iterating the successor function on the natural numbers. Numbers
are represented in mirrored base-2, i.e., with the \emph{leftmost}
digit the \emph{lowest}. Such a machine may be build with alphabet
$\{\triangleright,\mathtt{0},\mathtt{1},\mathtt{\sqcup}\}$, where
the symbol $\triangleright$ is used to mark the starting cell and
$\mathtt{\sqcup}$ is the \emph{blank }symbol. We omit the exact description
of the machine here. Of importance is the content of the tape: Omitting
blank ($\mathtt{\sqcup}$) cells, natural numbers are represented
as illustrated in Fig. \ref{fig1}.\vspace{-16pt} 
\begin{figure}
\begin{centering}
\begin{tabular}{c|c|c|c|c|c|c|c|c|c|c|c|c|c|c|c|c|c|c|c|c|c|c|}
\cline{2-3} \cline{5-7} \cline{9-12} \cline{14-17} \cline{19-23} 
$0:$\enskip{} & $\triangleright$ & $\mathtt{0}$ & \quad{}$2:$\enskip{} & $\triangleright$ & $\mathtt{0}$ & $\mathtt{1}$ & \quad{}$4:$\enskip{} & $\triangleright$ & $\mathtt{0}$ & $\mathtt{0}$ & $\mathtt{1}$ & \quad{}$6:$\enskip{} & $\triangleright$ & $\mathtt{0}$ & $\mathtt{1}$ & $\mathtt{1}$ & \quad{}$8:$\enskip{} & $\triangleright$ & $\mathtt{0}$ & $\mathtt{0}$ & $\mathtt{0}$ & $\mathtt{1}$\tabularnewline
\cline{2-3} \cline{5-7} \cline{9-12} \cline{14-17} \cline{19-23} 
$1:$\enskip{} & $\triangleright$ & $\mathtt{1}$ & \quad{}$3:$\enskip{} & $\triangleright$ & $\mathtt{1}$ & $\mathtt{1}$ & \quad{}$5:$\enskip{} & $\triangleright$ & $\mathtt{1}$ & $\mathtt{0}$ & $\mathtt{1}$ & \quad{}$7:$\enskip{} & $\triangleright$ & $\mathtt{1}$ & $\mathtt{1}$ & $\mathtt{1}$ & \quad{}$9:$\enskip{} & $\triangleright$ & $\mathtt{1}$ & $\mathtt{0}$ & $\mathtt{0}$ & $\mathtt{1}$\tabularnewline
\cline{2-3} \cline{5-7} \cline{9-12} \cline{14-17} \cline{19-23} 
\end{tabular}
\par\end{centering}
\caption{\label{fig1}Mirrored base-2 Turing tape representation of $0,..,9\in\mathbb{N}_{0}$,
blank cells omitted. Notice that the mirrored notation causes perpetual
change close to the start cell, $\triangleright$.}
\vspace{-16pt}
\end{figure}

\noindent Initiated with its read-write head on the cell with the
start symbol $\triangleright$ of a tape with content $n$, the machine
will go through a number of configurations before returning the read-write
head to the start cell with the tape now having content $n+1$. Auto-iterating,
the machine will thus, over time, produce a tape that will have contained
every natural number in order. 

This Turing machine may be emulated by a finite $\Sigma{\scriptstyle \Gamma}$
on a set $X$ cf. Lemma~\ref{lem:Any-Turing}. Omitting the details\footnote{The details differ depending on whether $\Sigma{\scriptstyle \Gamma}$
must be static, but non-Boolean for Prop. \ref{prop:NonBool-NoPost},
or Boolean, but non-static for Prop. \ref{prop:Bool-Post}. See resp.
\cite{KleinRendsvigTuring} and \cite{BolanderBirkegaard2011}.}, the idea is that the Turing tape, or a finite fragment, thereof
may be encoded as a pointed Kripke model: Each cell of the tape corresponds
to a state, with the cell's content encoded by additional structure,\footnote{\label{fn:The-exact-form}For Prop. \ref{prop:Bool-Post}, tape cell
content may be encoded using atomic propositions, changable through
postconditions, cf. \cite{BolanderBirkegaard2011}; for Prop. \ref{prop:NonBool-NoPost},
cell content is written by adding and removing additional states,
cf. \cite{KleinRendsvigTuring}.} which is modally expressible. By structuring the cell states with
two equivalence relations and atoms $u$ and $e$ true at cells with
odd (even) index respectively, (cf. Fig. \ref{fig:tape}), also the
position of a cell is expressible. The designated state corresponds
to the start cell, marked~$\triangleright$.

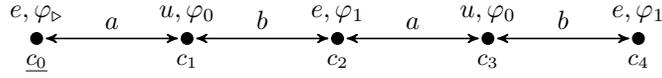
\begin{figure}
\begin{centering}
\scalebox{1}{
\begin{tikzpicture}[->,>=stealth',shorten >=1pt,shorten <=1pt,auto,node distance=2.8cm,semithick,label distance=0cm]

\tikzstyle{every state}=[draw=black,fill=black,inner sep=0,minimum size=.15cm]

    \node[state,,label=above:{$e,\varphi_{\triangleright}$},label=below:{\small $\textcolor{black}{\underline{c_{0}}}$}] at (-1,0) (0) {};
    \node[state,,label=above:{$u,\varphi_{0}$},label=below:{\small $\textcolor{black}{c_{1}}$}] at (1,0) (1) {};
    \node[state,,label=above:{$e,\varphi_{1}$},label=below:{\small $\textcolor{black}{c_{2}}$}] at (3,0) (2) {};
    \node[state,,label=above:{$u,\varphi_{0}$},label=below:{\small $\textcolor{black}{c_{3}}$}] at (5,0) (3) {};
 \node[state,,label=above:{$e,\varphi_{1}$},label=below:{\small $\textcolor{black}{c_{4}}$}] at (7,0) (4) {};

\path (0) edge[<->, draw=black] node [midway] {$a$}
 (1);
 \path (1) edge[<->, draw=black]node [midway] {$b$}(2);
 \path (2) edge[<->, draw=black] node [midway] {$a$}(3);
\path (3) edge[<->, draw=black] node [midway] {$b$}(4);

\end{tikzpicture}
}
\par\end{centering}
\caption{{\small{}\label{fig:tape}}A pointed Kripke model emulating the configuration
of the Turing machine with cell content representing the number 10.
The designated state is the underlined $c_{0}$. Each state is labeled
with a formula $\varphi_{\triangleright},\ \varphi_{0}$ or $\varphi_{1}$
expressing its content. Relations $a$ and $b$ allow expressing distance
of cells: That $c_{0}$ satisfies $\lozenge_{a}(u\wedge\lozenge_{b}(e\wedge\varphi_{1}))$
exactly expresses that cell $c_{2}$ contains a 1. Omitted are reflexive
loops for relations, and the additional structure marking cell content
and read-write head position.}
\vspace{-16pt}
\end{figure}

Let $(c_{n})_{n\in\mathbb{N}_{0}}$ be the sequence of configurations
of the machine when initiated on a tape with content 0. Each $c_{n}$
may be represented by a pointed Kripke model, obtaining a sequence
$(x_{n})_{n\in\mathbb{N}_{0}}$. By Lemma \ref{lem:Any-Turing}, there
thus exists a $\Sigma{\scriptstyle \Gamma}$ such that for all $n$,
$x_{n}\otimes\Sigma{\scriptstyle \Gamma}=x_{n+1}$. Hence, moving
to the full modal space $\boldsymbol{X}$ for the language used, a
clean map $\boldsymbol{f}\colon\boldsymbol{X}\rightarrow\boldsymbol{X}$
based on $\Sigma{\scriptstyle \Gamma}$ will satisfy $\boldsymbol{f}(\boldsymbol{x}_{n})=\boldsymbol{x}_{n+1}$
for all $n$. The Turing machine's run is thus emulated by $(\boldsymbol{f}^{k}(\boldsymbol{x}_{0}))_{k\in\mathbb{N}_{0}}$.

Let $(c'_{n})_{n\in\mathbb{N}_{0}}$ be the subsequence of $(c_{n})_{n\in\mathbb{N}_{0}}$
where the machine has finished the successor operation and returned
its read write head to its starting position $\triangleright$, ready
to embark on the next successor step. The tape of the first 9 of these
$c_{n}'$ are depicted in Fig. \ref{fig1}. Let $(\boldsymbol{x}'_{n})_{n\in\mathbb{N}_{0}}$
be the corresponding subsequence of $(\boldsymbol{f}^{k}(\boldsymbol{x}_{0}))_{k\in\mathbb{N}_{0}}$.
We show that $(\boldsymbol{x}'_{n})_{n\in\mathbb{N}_{0}}$ has uncountably
many limit points:

For each subset $Z$ of $\mathbb{N}$, let $c^{Z}$ be a tape with
content $1$ on cell $i$ iff $i\in Z$ and $0$ else. On the Kripke
model side, let the corresponding $x^{Z}\in X$ be a model structurally
identical to those of $(x'_{n})_{n\in\mathbb{N}_{0}}$, but satisfying
$\varphi_{1}$ on all ``cell states'' distance $i\in Z$ from the
designated ``$\triangleright$'' state, and $\varphi_{0}$ on all
other.\footnote{The exact form is straightforward from the constructions used in \cite{KleinRendsvigTuring}
and \cite{BolanderBirkegaard2011}.} The set $\{x^{Z}\colon Z\subseteq\mathbb{N}\}$ is uncountable, and
each $\boldsymbol{x^{Z}}$ is a limit point of $\overline{\boldsymbol{x}}$:
For each $Z\subseteq\mathbb{N}$ and $n\in\mathbb{N}$, there are
infinitely many $k$ for which $x_{k}\vDash\varphi$ iff $x^{Z}\vDash\varphi$
for all $\varphi$ of modal depth at most $n$. Hence, for every $n$,
the set $\{\boldsymbol{x}_{k}:d_{b}(\boldsymbol{x}_{k},\boldsymbol{x^{Z}})<2^{-n}\}$
is infinite, with $d_{b}$ the equivalent of the $n$-bisimulation
metric, cf. Prop. \ref{prop:nBisim-finite}. Hence, for each of the
uncountably many $Z\subseteq\mathbb{N}$, $\boldsymbol{x^{Z}}$ is
a limit point of the sequence~$\boldsymbol{\overline{x}}$.

Finally, every $\boldsymbol{x}'_{k}\in(\boldsymbol{x}'_{n})_{n\in\mathbb{N}_{0}}$
is recurrent: That $\boldsymbol{x}'_{k}\in\omega_{\boldsymbol{f}}(\boldsymbol{x}'_{k})$
follows from $\boldsymbol{x}'_{k}$ being a limit point of $(\boldsymbol{x}'_{n})_{n\in\mathbb{N}_{0}}$,
which it is as $\boldsymbol{x}'_{k}=\boldsymbol{x}^{Z}$ for some
$Z\subseteq\mathbb{N}$.\footnote{A similar argument shows that all $x^{Z}$ with $Z\subseteq\mathbb{N}$
co-infinite are recurrent points. Hence $\omega_{f}(\boldsymbol{x}'_{k})$
for any $\boldsymbol{x}'_{k}\in(\boldsymbol{x}'_{n})_{n\in\mathbb{N}_{0}}$
contains uncountably many recurrent points.} As the set of recurrent points is thus infinite, it cannot be periodic. \qed
\end{proof}

As a final result on the orbits of clean maps, we answer an open question:
After having exemplified a period 2 system, Sadzik \cite{Sadzik2006}
notes that it is unknown whether finite, static, but non-Boolean,
clean maps exhibiting longer periods exist. They do:
\begin{proposition}
\label{prop:nPeriod}For any $n\in\mathbb{N}$, there exists finite,
static, but non-Boolean clean maps with periodic orbits of period
$n$. This is also true for finite Boolean, but non-static, clean
maps.
\end{proposition}

\begin{proof}
For the given $n$, find a Turing machine that, from some configuration,
loops with period $n$. From here, Lemma \ref{lem:Any-Turing} does
the job.\qed
\end{proof}

Finally, we note that brute force determination of a clean map's orbit
properties is not in general a feasible option:
\begin{proposition}
The problems of determining whether a Boolean and non-static, or a
static and non-Boolean, clean map, a) has a periodic orbit or not,
and b) contains a limit cycle or not, are both undecidable.
\end{proposition}

\begin{proof}
The constructions from the proofs of Lemma \ref{lem:Any-Turing} allows
encoding the halting problem into either question.\qed
\end{proof}

\section{\label{sec:Concl}Discussion and Future Venues}

We consider Theorem \ref{thm:DynSys} our main contribution. With
it, an interface between the discrete semantics of dynamic epistemic
logic with dynamical systems have been provided; thus the former has
been situated in the mathematical field of the latter. This paves
the way for the application of results from dynamical systems theory
and related fields to the information dynamics of dynamic epistemic
logic.\smallskip{}

The term \emph{nontrivial recurrence} is adopted from Hasselblatt
and Katok, \cite{Hassellblatt2002}. They remark that ``{[}nontrivial
recurrence{]} is the first indication of complicated asymptotic behavior.''
Propositions \ref{prop:Bool-Post} and \ref{prop:NonBool-NoPost}
indicate that the dynamics of action models and product update may
not be an easy landscape to map. Hasselblatt and Katok continue: ``In
certain low-dimensional situations {[}...{]} it is possible to give
a comprehensive description of the nontrivial recurrence that can
appear.'' \cite[p. 24]{Hassellblatt2002}. That the Stone topology
is zero-dimensional fuels the hope that general topology and dynamical
systems theory yet has perspectives to offer on dynamic epistemic
logic. One possible direction is seeking a finer parametrization of
clean maps%
{} combined with results specific to zero-dimensional spaces, as found,
e.g., in the field of symbolic dynamics \cite{Lind1995}. But also
other venues are possible: The introduction of \cite{Hassellblatt2002}
is counts an inspiration.\smallskip{}

The approach presented furthermore applies to model transformations
beyond multi-pointed action models and product update. Given the equivalence
shown in \cite{KooiRenne2011} between single-pointed action model
product update and \emph{general arrow updates}, we see no reason
to suspect that ``clean maps'' based on the latter should not be
continuous on modal spaces. A further conjecture is that the \emph{action-priority
update} of \cite{Baltag_Smets_2008} on plausibility models\footnote{Hence also the multi-agent belief revision policies \emph{lexicographic
upgrade} and \emph{elite change}, also known as \emph{radical} and
\emph{conservative upgrade}, introduced in \cite{Benthem2007b}, cf.
\cite{Baltag_Smets_2008}. } yields ``clean maps'' continuous w.r.t. the suited Stone topology,
and that this may be shown using a variant of our proof of the continuity
of clean maps. A more difficult case is the \emph{PDL-transformations}
of \emph{General Dynamic Dynamic Logic} \cite{Girard2012_GDDL} given
the signature change the operation involves. \smallskip{}

There is a possible clinch between the suggested approach and epistemic
logic with common knowledge. The state space of a dynamical system
is compact. The Stone topology for languages including a common knowledge
operator is non-compact. Hence, it cannot constitute the space of
a dynamical system\textemdash but its \emph{one-point compactification}
may. We are currently working on this clinch, the consequences of
compactification, and relations to the problem of attaining common
knowledge, cf. \cite{HalpernMoses1990}.\smallskip{}

Questions also arise concerning the \emph{dynamic logic }of dynamic
epistemic logic. Propositions \ref{prop:NonBool-NoPost} and \ref{prop:Bool-Post}
indicate that there is more to the semantic dynamics of dynamic epistemic
logic than is representable by finite compositional dynamic modalities\textemdash even
when including a Kleene star. An open question still stands on how
to reason about limit behavior. One interesting venue stems from van
Benthem \cite{Benthem2011b}. He notes\footnote{In the omitted part of the quotation from the introduction.}
that the reduction axioms of dynamic epistemic logic could possibly
be viewed on par with differential equations of quantitative dynamical
systems. As modal spaces are zero-dimensional, they are imbeddable
in $\mathbb{R}$ cf. \cite[Thm 50.5]{Munkres}, turning clean maps
into functions from $\mathbb{R}$ to $\mathbb{R}$, possibly representable
as discrete-time difference equations. \smallskip{}

An alternative approach is possible given by consulting Theorem \ref{thm:DynSys}.
With Theorem \ref{thm:DynSys}, a connection arises between dynamic
epistemic logic and \emph{dynamic topological logic} (see e.g. \cite{Fernandez-Duque2012,Fernandez-Duque2012a,Kremer1997,KremerMints2007}):
Each system $(\boldsymbol{X}_{d},\boldsymbol{f})$ may be considered
a dynamic topological model with atom set $\boldsymbol{\mathcal{L}}_{\Lambda}$
and the `next' operator's semantics given by an application of $\boldsymbol{f}$,
equivalent to a $\langle\boldsymbol{f}\rangle$ dynamic modality of
DEL. The topological `interior' operator has yet no DEL parallel.
A `henceforth' operator allows for a limited characterization of
recurrence \cite{KremerMints2007}. We are wondering about and wandering
around the connections between a limit set operator with semantics
$x\vDash[\omega_{\boldsymbol{f}}]\varphi$ iff $y\vDash\varphi$ for
all $\boldsymbol{y}\in\omega_{\boldsymbol{f}}(\boldsymbol{x})$, dynamic
topological logic and the study of oscillations suggested by van Benthem
\cite{Benthem-DS-2016}.\smallskip{}

With the focal point on pointed Kripke models and action model transformations,
we have only considered a special case of logical dynamics. It is
our firm belief that much of the methodology here suggested is generalizable:
With structures described logically using a countable language, the
notion of logical convergence will coincide with topological convergence
in the Stone topology on the quotient space \emph{modulo} logical
equivalence, and the metrics introduced will, \emph{mutatis mutandis},
be applicable to said space \cite{BigTech}. The continuity of maps
and compactness of course depends on what the specifics of the chosen
model transformations and the compactness of the logic amount to.

\subsubsection*{Acknowledgements}

{\small{}The contribution of R.K. Rendsvig was funded by the Swedish
Research Council through the Knowledge in a Digital World project
and by The Center for Information and Bubble Studies, sponsored by
The Carlsberg Foundation. The contribution of D. Klein was partially
supported by the Deutsche Forschungsgemeinschaft (DFG) and Grantov\'{a}
agentura \v{C}esk\'{e} republiky (GA\v{C}R) as part of the joint project
From Shared Evidence to Group Attitudes {[}RO 4548/6-1{]}.}{\small \par}

\bibliographystyle{abbrv}

\end{document}